\documentclass[final,5p,times,twocolumn]{elsarticle}

\usepackage{amssymb}
\usepackage[cmex10]{amsmath}
\usepackage{amssymb,amsthm,alltt}
\usepackage[all]{xy}
\xyoption{poly}

\usepackage{stackrel}
\usepackage{xspace}
 \usepackage[usenames]{color}
\usepackage{graphics} %needed for the NEW quantifier
\usepackage{logic-article}
\usepackage[textsize=small]{todonotes}

\begin{document}

\begin{frontmatter}



\title{Equivariant algorithms for constraint satisfaction problems over coset templates\tnoteref{sosna}}
\tnotetext[sosna]{Supported by 
%ERC Starting Grant ``Sosna'' and by 
the NCN grant 2012/07/B/ST6/01497.}


\author{S{\l}awomir Lasota}
\address{University of Warsaw}

\begin{abstract}
We investigate the Constraint Satisfaction Problem (CSP) over templates with a group structure,
and algorithms solving CSP that are \emph{equivariant}, i.e.~invariant under a natural group action induced by a template.
Our main result is a method of proving the implication: if CSP over a coset template $T$ is solvable by a local equivariant algorithm then $T$ is 2-Helly (or equivalently, has a majority polymorphism). 
Therefore bounded width, and definability in fixed-point logics, coincide with 2-Helly.
Even if these facts may be derived from already known results, our new proof method has two advantages. 
First,  the proof is short, self-contained, and completely avoids referring to the omitting-types theorems.
Second, it brings to light some new connections between CSP theory and descriptive complexity theory, via
a construction generalizing CFI graphs. 
%
%We provide a short and elementary proof of the fact that bounded width implies strict width 2, for a subclass of CSP templates.
%Even if this fact may be easily derived from already known results, the new proof has three advantages. 
%First,  the proof  is short, self-contained, and completely avoids referring to the omitting-types theorems.
%Second, it brings to light some new connections between CSP theory and descriptive complexity theory, via
%a construction similar to CFI graphs. 
%Finally, the proof easily generalizes to show equivalence of all algorithms and logics that commute with the natural action
%induced by a template.
\end{abstract}


\end{frontmatter}


% !TEX root = main_bw_2Helly.tex

\newcommand{\slcomm}[1]{\todo[inline,color=blue!20]{#1}}

\newcounter{quotecount}
\newcommand{\MyQuote}[1]{\vspace{1cm}\refstepcounter{quotecount}%
     \parbox{7.3cm}{\em #1}\hspace*{.4cm}\normalfont(\arabic{quotecount})\\[1cm]}

\newcommand{\paragrafik}[1]{\vspace{0.1cm} \noindent {\bf #1.}}

\newcommand{\atoms}{\mathbb A}
\newcommand{\eqdef}{\stackrel {\text{def}} =}
\newcommand{\nat}{\mathbb{N}}
\newcommand{\aut}[2]{\text{Aut}({#2}/{#1})}
\newcommand{\pow}[1]{{\cal P}(#1)}

\newcommand{\restr}[2]{#1|_{#2}}
\newcommand{\dom}[1]{\text{dom}(#1)}
\newcommand{\presol}[1]{\text{Pre-sol}_{#1}}
\newcommand{\usun}[1]{}
\newcommand{\coset}[1]{[#1]}
\newcommand{\Z}{{\mathbb Z}}
\newcommand{\arity}[1]{\text{arity}(#1)}

\newcommand{\torus}{I}

% !TEX root = main_bw_2Helly.tex

\section{Introduction}

Many natural computational problems may be seen as instantiations of a generic framework called
\emph{constraint satisfaction problems} (CSP).
In a nutshell, a CSP is parametrized by a \emph{template}, a finite relational structure $T$; the CSP over $T$  
asks if a given relational structure $I$ over the same vocabulary as $T$ admits a homomorphism to $T$ 
(called a solution of $I$).
For every template $T$, the CSP over $T$ (denoted CSP($T$)) is always in NP;
a famous conjecture due to Feder and Vardi~\cite{FV98} says that for every template $T$, the CSP($T$) is 
either solvable in P, or NP-complete. 

We concentrate on \emph{coset templates} where, roughly speaking, both the carrier set and the relations have 
a group structure.
The coset templates are cores and admit a Malcev polymorphism, and are thus in P~\cite{F05,BD06}.
A coset template $T$ naturally induces a group action on (partial) solutions. If, roughly speaking, 
the induced group action can be extended
to the state space of an algorithm solving CTP($T$), and the algorithm execution is invariant under the group action,
we call the algorithm \emph{equivariant}. We investigate equivariant algorithms which are \emph{local}, i.e.~update only a bounded amount of data in every single step of execution.

A widely studied family of local equivariant algorithms is the \emph{local consistency algorithms} that compute
families of partial solutions of bounded size conforming to a local consistency condition.
Templates $T$ whose CSP($T$) is solvable by a local consistency algorithm are said to have \emph{bounded width}.
Another source of examples of local equivariant algorithms are logics (via their decision procedures); 
relevant logics for us will be fix-point extensions of first order logic, like LFP or IFP or IFP+C (IFP with counting quantifiers)~\cite{fmt-book}.
We say that CSP($T$) is definable in a logic if some formula of the logic defines the set of all solvable instances of CSP($T$).
% !!! It is well known that bounded width is equivalent to definability in Datalog~\cite{FV98}.

%An even more severe restriction of tractable templates is \emph{strict width 2}, i.e.~templates whose
%CSP is solvable by the $(2, 3)$-consistency algorithm and moreover all the solutions may be constructed in a greedy way
%from the family of partial solutions computed by the algorithm. 
%Both subclasses have multiple alternative characterizations; for instance, bounded width is equivalent to definability in Datalog, 
%while strict width 2 is equivalent to existence of majority polymorphism~\cite{FV98}.

Our technical contribution is the proof of the following implication: if  CSP($T$), for a coset template $T$, is solvable by a local equivariant algorithm then $T$ is \emph{2-Helly}. 
In consequence, all local equivariant algorithms that can capture 2-Helly templates are equally expressive.
The 2-Helly property says that for every partial solution $h$ of an instance $I$,
if $h$ does not extend to a solution of $I$ then 
the restriction of $h$ to some two elements of its domain does not either. %extend to a solution of $I$.
This is a robust property of templates with many equivalent characterizations (e.g.~strict width 2, or existence of 
a majority polymorphism)~\cite{FV98}. 
%Our proof method is however not specific for bounded width, and can be 
%straightforwardly adapted to any other algorithm which is \emph{local} and \emph{equivariant}, and therefore
%if CSP($T$) is solvable by a local equivariant algorithm then $T$ is 2-Helly.
As a corollary we obtain equivalence of the following conditions for coset templates:
%
%\begin{itemize}
(i) 2-Helly;
(ii) bounded width; and % solvability by a local consistency algorithm;
(iii) definability in fix-point extensions of first-order logic.
%\end{itemize}
%
The corollary is not a new result; equivalence of the first two conditions may be inferred e.g.~from Lemma~9 in~\cite{DL08}
(even for all core templates with a Malcev polymorphism), while equivalence of the last two ones follows from~\cite{ABD09} 
together with the results of~\cite{BK09} (cf.~also~\cite{BK14}).
All these results build on Tame Congruence Theory~\cite{HM88}, and their proofs are a detour through the 
deep omitting-type theorems, cf.~\cite{LZ07}.
Contrarily to this, our proof has an advantage of being short, elementary, and self-contained,
thus offering a direct insight into the problem.
%Furthermore, our result is more general than~\eqref{eq:cor}: it implies that
%\emph{all} local and equivariant algorithms that can capture 2-Helly templates are equally expressive, 
%namely they solve CSP($T$) for exactly the same coset templates $T$.

Finally, our proof brings to light interesting connections between the 
CSP theory and the descriptive complexity theory: % (cf.~for instance~\cite{grohebook}):
the crucial step of the proof is essentially based on a construction similar to 
\emph{CFI graphs}, the intricate construction of Cai, F{\"u}rer and Immerman~\cite{CFI92}.
CFI graphs have been designed to separate properties of relational structures decidable in polynomial time
from IFP+C. 
A similar construction has been used later in~\cite{BKLT13} to show lack of determination of Turing machines in sets with atoms~\cite{BKL11}.
The crucial step of our proof is actually a significant generalization of the construction of~\cite{BKLT13}.

%As another example of connection between CSP theory and descriptive complexity, one should mention a result of~\cite{ABD09} which, together with the
%results of~\cite{BK09} (cf.~also~\cite{BK14}), 
%imply that bounded width (i.e.~definability in Datalog) is equivalent to definability of the CSP problem in the logic IFP+C. 
%The proof again strongly relies on the algebraic setting of~\cite{HM88}; contrarily to this, our elementary proof may be easily adapted
%to show that definability in IFP+C is equivalent to strict width 2, and thus also to bounded width, for coset templates.

%The direct motivation for this work came from~\cite{KLOT14}, where CSP theory, and especially bounded width property, has been used for
%an effective characterization of those alphabets that admit determination of Turing machines with atoms. 
%As templates used in~\cite{KLOT14} are coset ones,
%our proof may be used to eliminate referring to the algebra in~\cite{KLOT14}, thus making the proof self-contained.

For completeness we mention a recent paper of Barto~\cite{B14} which announces the collapse of
bounded width hierarchy for \emph{all} templates: bounded width implies width $(2, 3)$, which is however weaker than 
2-Helly in general.

\section{Preliminaries}

\subsection{Constraint satisfaction problems}

A \emph{template} $T$ is a finite relational structure, i.e.~consists of a finite carrier set $T$ 
(denoted by the same symbol as a template) and a finite family of relations in $T$.
Each relation $R \subseteq T^n$ is of a specified arity, $\arity{R} = n$.
Let $T$ be fixed henceforth.

An instance $I$ over a template $T$ consists of a finite set $I$ of elements, 
and a finite set of \emph{constraints}.
A constraint, written $R(a_1, \ldots, a_n)$, is specified by a template relation $R$
and an $n$-tuple of elements of $I$,  where $\arity{R} = n$.

A partial function $h$ from $I$ to $T$, with $\{a_1, \ldots, a_n\} \subseteq \dom{h}$, 
\emph{satisfies} a constraint $R(a_1, \ldots, a_n)$ in $I$ when $R(h(a_1), \ldots, h(a_n))$ holds in $T$.
If $h$ satisfies all constraints in its domain, $h$ is a \emph{partial solution} of $I$,
and $h$ is a \emph{solution} when it is total.
By the size of a partial solution $h$ we mean the size of $\dom{h}$.
The constraint satisfaction problem over $T$, denoted CSP($T$), is a decision problem that asks if 
a given instance over $T$ has a solution.

There are many equivalent formulations of the problem. For instance, one can see $I$ and $T$ as
relational structures over the same vocabulary, and then CSP($T$) asks if there is a homomorphism 
from $I$ to $T$. 

%For $R \in \Sigma$, whenever a tuple of elements $a_1, \ldots, a_n$ of a structure $I$ is related by 
%the relation $R$, we write $R(a_1, \ldots, a_n)$ and call the tuple a \emph{constraint}. 
%In the sequel we will write $R$ for a symbol in $\Sigma$, as well as for a relation
%$R \subseteq I^n$ in a structure $I$, hoping that not confusion occurs.

\subsection{2-Helly templates}

%\begin{definition}
For an instance $I$ over some template, % with at least $k$ elements, 
and $k < j$,
a \emph{$(k, j)$-anomaly} is a partial solution $h$ of $I$ of size $j$
that does not extend to a solution, such that restriction of $h$ to every $k$-element subset of $\dom{h}$ does extend to a solution.
%\end{definition}
%
%\noindent
Clearly a $(k, j)$-anomaly is also $(k', j)$-anomaly, for $k' < k$.
\begin{definition} \label{def:2Helly}
A template $T$ is 2-Helly if no instance of $T$ admits a $(2, j)$-anomaly, for $j > 2$.
\end{definition}
In other words: for every partial solution $h$ of size $j > 2$,
if the restriction of $h$ to every 2-element subset of its domain extends to a solution then
$h$ does extend to a solution too. 
Analogously one may define $k$-Helly for arbitrary $k$, which however will not be needed here.

We conveniently characterize 2-Helly templates as follows.

\begin{lemma} \label{lem:iffanomaly}
A template $T$ is 2-Helly iff no instance of $T$ admits a $(k, k+1)$-anomaly, for $k \geq 2$.
% The arity of a template $T$ is equal to the greatest size of an anomaly (understood as the size of the domain) in an instance of $T$.
\end{lemma}
\begin{proof}
For one direction, we observe that a $(k, k+1)$-anomaly is also a $(2, k+1)$-anomaly.

For the other direction, consider an instance with some fixed $(2, j)$-anomaly $h$, for $j > 2$. 
For every subset $X \subseteq \dom{h}$, the restriction $\restr{h}{X}$ either extends to a solution of $I$, or not.
Consider the minimal subset $X$ wrt.~inclusion such that  $\restr{h}{X}$ does not extend to a solution of $I$.
For all strict subsets $X' \subseteq X$, $\restr{f}{X'}$ extends to a solution, hence $\restr{f}{X}$ is a
$(k-1, k)$-anomaly, where $k$ is the size of $X$. Note that $k > 2$.
%If $k$ is the cardinality of $X$
%For every $k \in \{2, \ldots, j-1\}$, and for every $(k+1)$-element subset of $\dom{h}$, restriction of $h$ to this subset either is a 
%$(k, k+1)$-anomaly or not.
%We claim that as long as $h$ is a $(2, j)$-anomaly, the restriction of $h$ to at least one of the subsets must be an $(k, k+1)$-anomaly as well.
\end{proof}

\subsection{The pp-definable relations}
We adopt the convention to mention explicitly the free variables of a formula $\phi(x_1, \ldots, x_n)$.
In the specific instances $I$ of CSP($T$) used in our proof it will be convenient to use \emph{pp-definable} relations, i.e.~relations 
definable by an existential first-order formula of the form:
\begin{align} \label{eq:ppformula}
\phi(x_1, \ldots, x_n) \ \equiv \ \exists x_{n+1}, \ldots, x_{n+m} . \ \psi_1 \land \ldots \land \psi_l,
\end{align}
where every subformula $\psi_i$ is an atomic proposition $R(x_{i_1}, \ldots, x_{i_j})$, for some
template relation $R$.
The formula $\phi$ defines the $n$-ary relation in $T$ containing the tuples
\[
(t_1, \ldots, t_n) \in T^n
\]
such that the valuation $x_1 \mapsto t_1, \ldots, x_n \mapsto t_n$ satisfies $\phi$.
% referring only to variables $x_1, \ldots, x_{n+m}$.
The pp-definable relations are closed under projection and intersection.

In the sequel we feel free to implicitly assume that elements of an instance are totally ordered.
The implicit order allows us to treat (partial) solutions as tuples, and allows to state the following useful fact:
\begin{fact} \label{fact:solutions pp-definable}
Let $X \subseteq I$ be a subset of an instance. The set of partial solutions with domain $X$ that extend to a solution of $I$,
if nonempty, is pp-definable.
\end{fact} 
% !!!!
%We recall the following widely known fact (see for instance~\cite{LZ07}):
%\begin{fact} \label{fact:pp-definable equivalent}
%Adding a pp-definable relation to a template preserves bounded width.
%% yields a computationally equivalent template. In particular, bounded width and 2-Helly are preserved.
%\end{fact}

\subsection{Almost-direct product of groups}  \label{sec:almost}

Overloading the notation, we write 1 for the identity element in any group. 
We use the diagrammatic order for writing the group operation $\tau \pi$ on elements $\tau, \pi$ of a group.

In the proof we will need the following elementary notion from group theory.\footnote{
The notion seems to be of independent interest; it is related to the \emph{arity} of a 
permutation group, 
as investigated for instance by Cherlin at al.~in~\cite{Cherlin96}.}
\begin{definition} \label{def:almost direct product}
Let $G_1$, $G_2$ and $G_3$ be arbitrary finite groups and let $H \leq G_1 \times G_2 \times G_3$ 
be a subgroup of the direct product. We call $H$ an \emph{almost-direct product} of $G_1, G_2, G_3$
if $H$ verifies the following conditions:
\begin{align} 
\label{eq:stronganomaly1}
%\hspace{-0.5cm}
%\exists \pi_1 \in S_1, \ \pi_2 \in S_2, \ \pi_3 \in S_3 . \ \neg R(\pi_1, \pi_2, \pi_3) \\
H \neq G_1 \times G_2 \times G_3  \\
\label{eq:stronganomaly2}
\forall \pi_2 \in G_2, \ \pi_3 \in G_3, \ \exists \pi_1 \in G_1 . \ (\pi_1, \pi_2, \pi_3) \in H \\
\label{eq:stronganomaly3}
\forall \pi_1 \in G_1, \ \pi_3 \in G_3, \ \exists \pi_2 \in G_2 . \ (\pi_1, \pi_2, \pi_3) \in H \\
\label{eq:stronganomaly4}
\forall \pi_1 \in G_1, \ \pi_2 \in G_2, \ \exists \pi_3 \in G_3 . \ (\pi_1, \pi_2, \pi_3) \in H 
\end{align}
Furthermore, an almost-direct product $H$ is \emph{strict} if $\pi_1$ (resp.~$\pi_2$, $\pi_3$)
in condition~\eqref{eq:stronganomaly2} (resp.~\eqref{eq:stronganomaly3}, \eqref{eq:stronganomaly4})
is uniquely determined.
%there is exactly one $\pi_1$ 
%satisfying~\eqref{eq:stronganomaly2}, and likewise for $\pi_2$ and $\pi_3$ in~\eqref{eq:stronganomaly3} 
%and~\eqref{eq:stronganomaly4}. 
\end{definition}
Let $H \leq G_1 \times G_2 \times G_3$ be an almost-direct product.
Consider the following normal subgroup $N_1$ of $G_1$:
\[
N_1 \ = \ \{ \pi_1 \in G_1 \ : \ (\pi_1, 1, 1) \in H \}.
\]
% and observe that $N_1$ is a normal subgroup of $G_1$.
Likewise define the normal subgroups $N_2$ and $N_3$ of $G_2$ and $G_3$, respectively.
In consequence, the product $N = N_1 \times N_2 \times N_3$ is 
% not just a subgroup of $H$ (which follows by the very definition of $N_1, N_2, N_3$), but actually 
a normal subgroup of $H$.  
Define the groups $\coset{G_1}, \coset{G_2}, \coset{G_3}$ and $\coset{H}$ as the quotients by
$N_1, N_2, N_3$ and $N$, respectively. 

By the definition of $N_1$, the quotient group $\coset{G_1}$ is obtained by identifying
its elements $\pi_1, \pi'_1$ that are equivalent:
\[
\pi_1 \equiv_1 \pi'_1 \iff (\forall \pi_2, \pi_3, (\pi_1, \pi_2, \pi_3) \in H \iff (\pi'_1, \pi_2, \pi_3) \in H).
\]
Similarly one defines the equivalences $\equiv_2$ and $\equiv_3$.
Note that $H$ is closed under the three equivalences; for instance,
\begin{align}  \label{eq:equiv}
(\pi_1, \pi_2, \pi_3) \in H \ \text{ and } \ \pi_1 \equiv_1 \pi'_1 \implies (\pi'_1, \pi_2, \pi_3) \in H.
\end{align}

\begin{lemma} \label{lem:strictquotient}
The quotient group $\coset{H}$ is a strict almost-direct product of $\coset{G_1}, \coset{G_2}, \coset{G_3}$.
%Every almost-direct product $H$ has a surjective homomorphism onto a strict almost-direct product.
\end{lemma}
\begin{proof}

$\coset{H}$, being the quotient of $H$, is an almost-direct product of $\coset{G_1}, \coset{G_2}, \coset{G_3}$.
We claim that $\coset H$ is strict.
Concentrating on point~\eqref{eq:stronganomaly2} in 
Definition~\ref{def:almost direct product} (the remaining two conditions are treated similarly),
we need to prove uniqueness of $\pi_1$.
Suppose  
\[
(\pi_1, \pi_2, \pi_3) \in \coset{H} \quad \text{ and } \quad
(\pi'_1, \pi_2, \pi_3) \in \coset{H};
\]
we need to derive $\pi_1 = \pi'_1$.
As $H$ is closed under the three equivalences, there are some
$\rho_1, \rho'_1, \rho_2, \rho_3$
such that
\[
\tau = (\rho_1, \rho_2, \rho_3) \in H, \qquad
\tau' = (\rho'_1, \rho_2, \rho_3) \in H,
\]
and  (writing $\coset{\rho}$ for the equivalence class containing $\rho$)
\[
\coset{\rho_1} = \pi_1, \quad \coset{\rho'_1} = \pi'_1, \quad \coset{\rho_2} = \pi_2, \quad \coset{\rho_3} = \pi_3;
\]
and we need to derive $\rho_1 \equiv_1 \rho'_1$. The equivalence follows easily:
whenever $\sigma = (\pi_1, \tau_2, \tau_3) \in H$, we have
\[
(\pi'_1, \tau_2, \tau_3) = \sigma \, \tau^{-1} \, \tau' \in H.
\]
%\begin{align} \label{eq:quot}
%(\rho_1, \rho_2, \rho_3) \in H \quad \text{ and } \quad (\tau_1, \tau_2, \tau_3) \in H,
%\end{align}
%and therefore (we write $\coset{\rho}$ for the coset containing $\rho$) 
%\[
%(\coset{\rho_1}, \coset{\rho_2}, \coset{\rho_3}) \in \coset{H} \quad \text{ and } \quad
%(\coset{\tau_1}, \coset{\tau_2}, \coset{\tau_3}) \in \coset{H},
%\]
%and suppose $\coset{\rho_2} = \coset{\tau_2}$ and $\coset{\rho_3} = \coset{\tau_3}$. 
%We will demonstrate $\coset{\rho_1} \ = \ \coset{\tau_1}$.
%
%Using the definition of $N_2$ and $\coset{\rho_2} = \coset{\tau_2}$, we deduce
%that  for all $\pi_1 \in G_1$, $\pi_3 \in G_3$,
%\[
%(\pi_1, \rho_2, \pi_3) \in H \text{ if and only if }
%(\pi_1, \tau_2, \pi_3) \in H.
%\]
%Applying this equivalence to the first half of~\eqref{eq:quot} we obtain
%\[
%(\rho_1, \tau_2, \rho_3) \in H.
%\]
%Using $\coset{\pi_3} = \coset{\tau_3}$, from the above fact we similarly deduce 
%\[
%(\rho_1, \tau_2, \tau_3) \in H
%\]
%which shows, together with the second half of~\eqref{eq:quot}, that
%\[
%\coset{\pi_1} \ = \ \coset{\tau_1}
%\]
%thus proving the uniqueness of $\pi_1$ in condition~\eqref{eq:stronganomaly2}.
\end{proof}
\begin{lemma} \label{lem:abelian}
A strict almost-direct product is commutative.%\footnote{I am grateful to Szymek Toru{\'n}czyk for the simplification of the proof.}
\end{lemma}
\begin{proof}
Let $H \leq G_1 \times G_2 \times G_3$ be a strict almost-direct product and let
$\pi, \tau \in G_1$.
We know that there exist $\rho_2 \in G_2$ and $\rho_3 \in G_3$ so that (we do not use the uniqueness of $\rho_2$ and $\rho_3$ here):
\[
(\pi, 1, \rho_3) \in H \quad \text{ and } \quad (\tau, \rho_2, 1) \in H.
\]
Applying the group operation to these two elements in two different orders we get:
\[
(\pi \tau, \rho_2, \rho_3)  \in H \quad \text{ and } \quad (\tau \pi, \rho_2, \rho_3) \in H.
\]
Now using the uniqueness of $\pi \tau$ (and $\tau \pi)$, we deduce that $\pi \tau = \tau \pi$. As $\pi$ and $\tau$ have been chosen
arbitrarily, the group $G_1$ is commutative. Likewise for $G_2$ and $G_3$, and in consequence also for
the subgroup $H \leq G_1 \times G_2 \times G_3$.%\footnote{We are grateful to Szymon Toru{\'n}czyk for providing that simple proof of this fact.}
\end{proof}

\section{Coset templates}
Below by a coset we always mean a right coset. (This choice is however arbitrary and we could consider left cosets instead.)
 
\begin{definition}
\emph{Coset templates} are particular templates $T$ that satisfy the following conditions:
\begin{itemize}
\item the carrier set of $T$ is a disjoint union of groups, call these groups \emph{carrier groups};
\item every $n$-ary relation $R$ in $T$ is a coset in the direct product $G_1 \times \ldots \times G_n$ 
of some carrier groups $G_1, \ldots, G_n$;
\item for a relation $R \subseteq G_1 \times \ldots \times G_n$ in $T$ and $\pi \in G_1 \times \ldots \times G_n$, 
the coset $R  \pi$ is also a relation in $T$;
\item for every carrier group $G$, $T$ has a unary relation $\{1\}$ containing exactly one element, the identity of $G$.  
\end{itemize}
\end{definition}
\noindent
Note that the last two conditions imply that a coset template contains every singleton as a unary relation, and thus is a 
rigid core, i.e.~admits no nontrivial endomorphisms.
\begin{example}
Here is a family of coset templates $T_n$, for $n \geq 2$. 
The carrier set of $T_n$ is $\{1, \pi\}$, the cyclic group of order 2.
Relations of $T_n$ are, except the two singleton unary relations $1(\_)$ and $\pi(\_)$, two $n$-ary relations
\[
R_\text{even}, \ R_\text{odd} \ \subseteq \ T^n
%\{ (0,0), (1,1) \} \quad \text{ and } \quad \{ (0,1), (1,0) \}.
\]
containing $n$-tuples where $\pi$ appears an even (resp.~odd) number of times.
Template $T_2$ is 2-Helly, while for $n > 2$, template $T_n$ is not. Indeed, a (2,3)-anomaly is admitted by 
an instance over $T_3$, consisting of three elements $a_1, a_2, a_3$ and four constraints:
\[
\pi(a_1) \qquad \pi(a_2) \qquad \pi(a_3) \qquad R_\text{even}(a_1, a_2, a_3).
\]
\end{example}

Consider a relation $R \subseteq G_1 \times \ldots \times G_n$ in a coset template, and an instance $I$.
For a constraint $R(a_1, \ldots, a_n)$ in $I$ and $i \in \{1 \ldots n\}$, we call $G_i$ a \emph{constraining group} of $a_i$.
In order to have a solution, an instance $I$ has to be non-contradictory, in the sense that every element must have exactly one constraining group
(elements with no constraining group may be safely removed from $I$).
We only consider non-contradictory instances from now on.

Consider a fixed coset template $T$ and an instance $I$ over $T$.
By a \emph{pre-solution} of $I$ we mean any function $s : I \to T$ that maps every element $i \in I$ to an element of the constraining group of $i$. 
Pre-solutions of an instance $I$ form a group, % $\presol{I}$, 
with group operation defined point-wise.
One can also speak of partial pre-solutions, whose domain is a subset of $I$.
Using an implicit order of elements of an instance, 
(partial) pre-solutions of $I$ are elements of the direct product of constraining groups of (some) elements of $I$. 

We distinguish \emph{subgroup instances}, where all relations $R$ appearing in the constraints $R(a_1, \ldots, a_n)$ 
are subgroups, instead of arbitrary cosets.
\begin{fact}  \label{fact:solutions coset}
%Fix a coset template.
(1) The set of all solutions $\cal H$ of an instance $I$, if nonempty, is a coset in the group of pre-solutions.
%$\presol{I}$.  %the direct product of constraining groups.
(2) In consequence, if $I$ is a subgroup instance then $\cal H$ is a subgroup of the group of pre-solutions.
% $\presol{I}$.  % the direct product of constraining groups.
\end{fact}
\begin{proof}
To show (1)
observe that for every constraint $c = R(a_1, \ldots, a_n)$ in $I$, the set of all pre-solutions ${\cal H}_c$ satisfying that particular constraint is a coset in the group of pre-solutions.
As solutions ${\cal H}$ are exactly the intersection,
\[
{\cal H} \ = \ \bigcap_c \ {\cal H}_c,
\]
for $c$ ranging over all constraints in $I$, 
by closure of cosets under nonempty intersection we derive that $\cal H$ is a coset.

(2) follows by an observation that the tuple $(1, \ldots, 1)$ of identities is always a solution, in case of
a subgroup instance. 
\end{proof}

%Wlog.~we may assume that every variable in a formula~\eqref{eq:ppformula} specifying a pp-definable set,
%appears in some atomic proposition. % (intuitively, the variable is constrained).
%Then 
Every pp-definable relation is essentially a projection of the set of solutions of some instance
(variables are element of the instance, and atomic propositions are its constraints), and
by Fact~\ref{fact:solutions coset} we derive the following corollary:
\begin{fact} \label{fact:action ppdef}
(1) Every pp-definable relation $R \subseteq G_1 \times \ldots \times G_n$ in $T$ 
is a coset in $G_1 \times \ldots \times G_n$.
(2) If $R$ is pp-definable and $\pi \in G_1 \times \ldots \times G_n$ then $R \, \pi$ is pp-definable as well.
\end{fact}

We will later exploit the following property of coset templates:
\begin{lemma} \label{lem:anomaly}
If some subgroup instance admits a $(k,k+1)$-anomaly, for $k \geq 2$, 
then some subgroup instance admits a $(k-1, k)$-anomaly.
\end{lemma}
%
% \noindent Hence some instance admits a $(2, 3)$-anomaly.
%
\begin{proof} % [Proof of Lemma~\ref{lem:anomaly}]
Fix a $(k, k+1)$-anomaly $h$ in a subgroup instance $I$, for some $k \geq 2$, and choose an arbitrary element $a \in \dom{h}$.
Let $X = \dom{h} \setminus\{a\}$. 
Define the new instance $I'$, with the same domain as $I$, whose constraints are all constraints of $I$ 
plus one additional unary constraint $1(a)$ requiring that $a$ should be mapped to the identity in its constraining group. 

As $h$ is an anomaly, the restriction $\restr{h}{\{a\}}$ extends to a solution of $I$, i.e.~$I$ has a solution $\bar h$ satisfying
${\bar h}(a) = h(a)$. 
Using an arbitrary such solution we define another partial solution $h'$ of $I$ with $\dom{h'} = \dom{h} = X\cup\{a\}$:
\[
h'(x) \ = \ h(x) \cdot \bar h^{-1}(x), \quad \text{ for } x\in X\cup\{a\}. 
\]
% As $\bar h(a) = 1$, $h'$ is also a partial solution of $I'$.
%As $h$ is a $(k, k+1)$-anomaly in $I$, using Fact~\ref{fact:action solution} we conclude that 
%$h'$ is also a $(k, k+1)$-anomaly in $I$.
%
Consider the restriction $h'' = \restr{h'}{X}$.
We claim that $h''$ is a $(k-1, k)$-anomaly in $I'$.
Indeed, for every subset $X' \subseteq X$ of size $k-1$, 
$\restr{h}{X'\cup\{a\}}$ extends to a solution of $I$, hence
$\restr{h'}{X'\cup\{a\}}$ extends to a solution of $I'$, 
and hence $\restr{h''}{X'}$ also extends to a solution of $I'$.
Moreover, $h$ does not extend to a solution of $I$, hence
$h'$ does not extend to a solution of $I'$, and thus $h''$ also does not extend to a solution of $I'$,
as every solution of $I'$ is forced to map $a$ to $1$.
\end{proof}

\subsection{Action of pre-solutions}

For a fixed instance $I$, 
define the action of pre-solutions on (partial) pre-solutions (thus in particular on (partial) solutions).
For a (partial) pre-solution $h : I \to T$ and a pre-solution $s$, let $h \cdot s$ be defined by the point-wise group operation:
\[
(h \cdot s)(a) = h(a) \, s(a), \quad \text{ for } a \in \dom{h}.
\]
We will apply the action to the instance $I$ itself: let
$
I \cdot s
$
be an instance with the same carrier set as $I$, whose constraints are obtained from the constraints of $I$ as follows:
for every constraint $R(a_1, \ldots, a_n)$ of $I$, 
the instance $I \cdot s$ contains a constraint
\[
(R \pi)(a_1, \ldots, a_n), \qquad \text{ where } \pi = (s(a_1), \ldots, s(a_n)).
\]
%As $R$ is a coset in the direct product of constraining groups of $a_1, \ldots, a_n$, so is $R'$.
Note that the action preserves constraining groups, and hence pre-solutions, of an instance.

It is important to notice that solvability is invariant under the action of pre-solutions: 
\begin{fact} \label{fact:action solution}
If h is a solution of $I$ then $h \cdot s$ is a solution of $I \cdot s$.
\end{fact}

\section{Local equivariant algorithms}   %The result}

\newcommand{\data}{{\cal D}}
\newcommand{\supp}{{\cal S}}

In the following we consider deterministic algorithms which run in \emph{stages}, and in every $i$th stage
a new object $\data_i(I)$ is computed as a function of the instance $I$ and previously computed objects $\data_1(I), \ldots, \data_{i-1}(I)$.
Thus an execution of an algorithm can be described as a sequence of $n(I)$ objects
\[
\data_1(I), \ \data_2(I),  \ \ldots, \ \data_{n(I)}(I).
\]
The outcome of an algorithm is determined by the final object $\data_{n(I)}(I)$.

We assume some action of pre-solutions $s$ of $I$ on the objects $\data_i(I)$, written $\data_i(I) \cdot s$.
An algorithm is called \emph{equivariant} when it commutes with the action:  for every pre-solution $s$,
\begin{align*}
n(I\cdot s) \ & = \ n(I) \\
\data_i(I \cdot s) \ & = \ \data_i(I) \cdot s, \ \text{ for every } i\leq n(I).
\end{align*}
%if, given an instance $I$, the algorithm reaches a global state ${\cal S} = (\data_1, \ldots, \data_i)$ via a sequence of stages,
%then for every pre-solution $s$ of $I$, given the instance $I \cdot s$, the algorithm reaches
%the global state ${\cal S} \cdot s = (\data_1 \cdot s, \ldots, \data_i \cdot s)$ via the same number of stages:
%\[
%\xymatrix{ {I} \ar@{|.>}[rr] \ar@{|->}[d]^s && {\cal S} \ar@{|->}[d]^s \\
%{I} \cdot s \ar@{|.>}[rr] && {\cal S} \cdot s
%}
%\]
%As solvability is invariant under action of pre-solutions, 
%a global state $\cal S$ of a correct equivariant algorithm is accepting (rejecting) if and only if the global state ${\cal S} \cdot s$ is so.
%
We say that $s$ \emph{fixes} $\supp \subseteq I$ when $s(x) = 1$ for all $x\in \supp$.
An equivariant algorithm is \emph{local} if there is a locality bound $d\in\nat$ independent from an instance,  such that
for every instance $I$ and $i \leq n(I)$, either (L0) holds, or both (L1) and (L2) hold:
\begin{itemize}
\item[(L0)] There is a subset $\supp \subseteq I$ of size at most $d$ such that $\data_i(I)$ depends only on the restriction of $I$ to $\supp$.
\end{itemize}
\begin{itemize}
\item[(L1)] $\data_i(I)$ depends only on at most $d$ objects among $\data_1(I), \ldots, \data_{i-1}(I)$.
%\vspace{-2mm} 
\item[(L2)] There is a subset $\supp \subseteq I$ of size at most $d$ such that $\data_i(I) \cdot s = \data_i(I)$ for every pre-solution $s$ fixing $\supp$. 
\end{itemize} 

The last condition (L2) is motivated by sets with atoms~\cite{BKL11,BKLT13} -- it corresponds to bounded support.
Roughly speaking, (L2) says that $\data_i(I)$ is only related to a bounded number of elements of $I$.
%We refrain from a further formalization of algorithms.  
For illustration, we demonstrate that local consistency algorithms, as well as the decision procedures of fix-point extensions of first-order logic, are equivariant and local.

\subsection{Local consistency algorithms}
Consider an instance $I$ and a family $\cal H$ of its partial solutions of size at most $k$, for some $k > 0$.
It will be convenient to split $\cal H$ into the subfamilies ${\cal H}_X$, 
where ${\cal H}_X = \{ h \in {\cal H} : \dom{h} = X\}$.
Fix $k$ and  $l \geq k$, and consider two subsets $X \subseteq Y$ of $I$ of size $k$ and $l$, respectively.
A partial solution $h \in {\cal H}_X$ is \emph{consistent} wrt.~${\cal H}$ and $(X, Y)$ if
\begin{quote}
$h$ extends to a partial solution $h'$ with $\dom{h'} = Y$,
whose restriction $\restr{h'}{X'}$ to every subset $X' \subseteq Y$ of size at most $k$ belongs to ${\cal H}_{X'}$.
\end{quote}
The \emph{$(k,l)$-consistency} algorithm takes as input an instance $I$ over $T$, and 
 computes the greatest family $\cal H$ of partial solutions of size at most $k$, such that 
every $h \in {\cal H}_X$  is consistent wrt.~$\cal H$ and $(X, Y)$, for every $X$ and $Y$ as above.
The algorithm starts with $\cal H$ containing all partial solutions of size $k$, and
 proceeds by iteratively removing from $\cal H$ all partial solutions  
$h$ that falsify the consistency condition.
The order of removing is irrelevant, but in order to guarantee equivariance we assume some fixed enumeration
$
(X_1, Y_1), \ldots, (X_n, Y_n)
$
of all pairs $(X, Y)$ of subsets of $I$ as above, that only depends on the size of the input $I$, but not on the constraints in $I$, 
and that the algorithm proceeds by iteratively executing the following subroutine until stabilization:
\begin{align*}
& \text{for } i = 1, 2, \ldots, n,\\
& \quad {\cal H}_{X_i} \ := \ \{ h \in {\cal H}_{X_i} : h \text{ is consistent wrt.~} H \text{ and } (X_i, Y_i) \} 
\qquad (*)
%\label{eq:stage}
\end{align*}
Every single update $(*)$ of ${\cal H}_{X_i}$, for some pair $(X_i, Y_i)$, constitutes a stage of the algorithm.

%Each stage applies to a subset $X$ of an input instance of size at most $k$, 
%and a superset $Y \supseteq X$ of size $l$, and evaluates the consistency condition for the pair $(X, Y)$.
%As the result of a stage, all partial solution $h \in {\cal H}_X$ falsifying the consistency condition for $(X,Y)$ are
%removed from ${\cal H}_X$.
%%\[ \{ h \in {\cal H}_X \ : \ (h, Y) \text{ satisfies the } (k, l)\text{-consistency} \} \]

When the stabilization is reached and $\cal H$ is nonempty then all the subfamilies ${\cal H}_X$ are also nonempty.
The \emph{$(k,l)$-consistency} algorithm accepts if the family $\cal H$ computed by the algorithm is nonempty; otherwise, the algorithm rejects.

%If an input instance $I$ is solvable (i.e.~admits a solution) then the $(k, l)$-consistency algorithm accepts, for every $k \leq l$.
%We say that a template $T$ has width $(k, l)$ if  
%the $(k, l)$-consistency algorithm correctly solves CSP($T$), % checks for existence of a solution on all instances over $T$, 
%namely the algorithm answers positively only if the input instance $I$ is solvable.
%$T$ has \emph{bounded width} if it has some width $(k, l)$.
%
%Observe that the $(k, l)$-consistency algorithm ignores relations in $T$ of arity greater than $k$.
%Therefore, if $T$ has width $(k, l)$ then wlog.~one can assume that 
%the arities of relations in $\Sigma$ are at most $k$.

%\emph{Strict width 2} is a stenghtening of width $(2, 3)$.
%Consider the family $\cal H$ computed by the $(2, 3)$-consistency algorithm.
%Call a partial function  $h : I \to T$ \emph{$\cal H$-consistent} if for every $\{a, b\} \subseteq \dom{h}$, $\restr{h}{\{a, b\}} \in {\cal H}_{\{a, b\}}$.
%In particular, every solution is consistent, but not every partial solution is so in general.
%A template $T$ has strict width 2 if for every $I$ for which  the $(2, 3)$-consistency algorithm computes 
%a nonempty family $\cal H$, 
%every $\cal H$-consistent partial function $I \to T$ extends to a solution of $I$.
%(The intuition is as follows: a solution can be build in a "greedy" way, from the empty partial function, 
%by iteratively extending the domain of the function with one point, preserving the invariant that
%the function is $\cal H$-consistent.)
%Clearly strict width 2 implies width $(2, 3)$.

Lift the action of pre-solutions to subfamilies ${\cal H}_X$ by direct image: ${\cal H}_X \cdot s = \{ h \cdot s : h \in {\cal H}_X\}$.
The $(k, l)$-consistency algorithm is equivariant then:
writing ${\cal H}_X(I, i)$ for the value of ${\cal H}_X$ after the $i$th stage of execution on
an instance $I$, we have:
\[
%\forall \ X, \quad 
{\cal H}_X(I \cdot s, i) \ = \ {\cal H}_X(I, i) \cdot s,
\]
for every $I$, every its pre-solution $s$, every $X\subseteq I$ and every $i$.
Furthermore, the $(k, l)$-consistency algorithm can be easily turned into a local one. Initially, subfamilies ${\cal H}_X$ satisfy (L0) for $d = k$. Otherwise, 
(L1) clearly holds, as the new value of ${\cal H}_{X_i}$ computed in one stage $(*)$ depends only on
\begin{align*} \label{eq:depends on}
%o(k, l) \ = \  
d \ = \ {l \choose k} \ + \ {l \choose {k-1}} \ + \ \ldots \ + \ {l \choose 1}
\end{align*}
values of ${\cal H}_{X'}$, for subsets $X' \subseteq Y_i$ of size at most $k$.  % of some superset $Y$ of $X$ of size $l$.
(L2) holds too: ${\cal H}_X(I, i) \cdot s \ = \ {\cal H}_X(I, i)$ for all pre-solutions fixing $X$.

\subsection{Fix-point logics}

There are many fix-point extensions of first-order logic. The logic LFP offers a construct of the least fix-point of a function definable by a formula. Here is an example formula:
\[
%\forall u, v \ [
\phi(u, v) \ \equiv \ \text{\sc LFP}_{R, x, y} \big[ E(x, y) \lor \exists z\  (E(x, z) \land R(z, y)) \big](u, v).
%](u, v)
\]
The formula has two free variables $u, v$ and defines the transitive closure of a binary relation $E$. 
As a further example, the formula $\forall x, y \ \phi(x, y)$ defines strong connectedness.

Evaluation of a formula of the form $\text{\sc LFP}_{R, \vec x}$ amounts to an iterative computation of the set of valuations of the variables $\vec x$, starting from the empty set of valuations, until stabilization. Given an arbitrary LFP formula $\phi$, a set of valuations is to be computed for every subformula of $\phi$. This can be turned into a stage-based local algorithm as follows. Let $\phi$ be a fixed LFP formula. The algorithm computes 
the sets ${\cal H}_X$, indexed by finite tuples $X \in I^*$ of elements of an instance $I$, such that for each $X = (a_1, \ldots, a_n)$, the set ${\cal H}_X$ contains a set of subformulas $\psi$ of $\phi$ for which $\psi(a_1, \ldots, a_n)$ holds.  The length of tuples $X$ is bounded by the greatest number of free variables of a subformula $\phi$; and every update of a set ${\cal H}_X$ only depends on a bounded number of other sets.
Therefore, the decision procedure for $\phi$ is local. It is also equivariant, as renaming the relations in $I$ into relations in $I\cdot s$ does not affect the iterative computation.

In the similar vein one argues that more expressive logics, like IFP (where the computation of fix-points is performed in the inflationary manner) or IFP+C (extension of IFP with counting), yield local and equivariant decision procedures as well.

\subsection{Non-equivariant algorithms}

As expected, many algorithms fail to satisfy equivariance.
As a first example, consider a naive ineffective algorithm that enumerates all pre-solutions $h$ and tests each for being a solution.
Enumerating and processing pre-solutions can be performed element-wise, thus leaving a hope for locality.
However, equivariance is violated. Indeed, suppose that on an instance $I$, the values $h(x)$ of pre-solution $h$ for an element $x\in I$ are enumerated in the order $\pi_1, \pi_2, \ldots$; then on an instance $I\cdot s$, the values $h(x)$ would have to be enumerated in the different order $\pi_1 \cdot s, \pi_2 \cdot s, \ldots$, which is not the case.

A similar phenomenon emerges in the polynomial-time algorithm for solving CSP over a template admitting a Malcev polymorphism, designed in~\cite{BD06} (or its generalized variant~\cite{dalmau}). The algorithm is applicable to coset templates as they admit a Malcev polymorphism, defined as $\phi(x, y, z) = x y^{-1} z$ (whenever $x, y, z$ are elements of the same carrier group). Roughly speaking, the algorithm manipulates a set of pre-solutions of an instance (succinctly represented by polynomially many representatives). For some fixed ordering $c_1, \ldots, c_n$ of all constraints in an instance, in its $k$th phase the algorithm computes the pre-solutions satisfying the constraints $c_1, \ldots, c_k$. 
Even if the core operation performed by the algorithm, namely computation of the Malcev polymorphism $\phi$, is equivariant, that is
\[
\phi(x \cdot s, y \cdot s, z \cdot s) \ = \ \phi(x, y, z) \cdot s,
\]
the whole algorithm is not so. 
As above, responsible for non-equivariance is enumeration of all elements of the template.

%On the other hand, if processing of total pre-solutions was split into coordinate-wise processing of small partial ones, 
%then locality would be recored, but equivariance would be sacrificed:

%
%a fixed formula  $\phi \equiv \exists X \ \psi(X)$ in existential monadic second order logic.
%Suppose that the decision procedure for $\phi$ enumerates all subsets of the instance $I$.
%As the size of subsets is not bounded, the algorithm violates locality. If, instead, manipulation of subsets is done in a point-wise manner, the algorithm can be turned in a local one which, however, fails to be equivariant in general.
%%
%Thus, depending on the 'granularity' of computations, the decision procedure would violate either locality or equivariance.

\section{Local equivariant algorithm implies 2-Helly}
% Proof of Theorem~\ref{thm:bw2sw2}}  % Bounded width implies 2-Helly} 
\label{sec:proof}

%As our main technical result we prove:
%
\begin{theorem}  \label{thm:bw2sw2}
For a coset template $T$, if CSP($T$) is solvable by a local equivariant algorithm then
% bounded width implies (and is thus equivalent to) 
$T$ is 2-Helly.
\end{theorem}
In other words, no local equivariant algorithm can solve CSP($T$), when $T$ a coset template but not 2-Helly.
As a direct corollary, bounded width implies 2-Helly for coset templates.
Note that the converse of Theorem~\ref{thm:bw2sw2} holds as well, as 2-Helly implies bounded width.

Another direct consequence of Theorem~\ref{thm:bw2sw2} is that
coset templates $T$ for which CSP($T$) is definable in LFP, IFP or IFP+C, are 2-Helly.
In consequence, over coset templates, all the mentioned fix-point extensions of first-order logic are equally expressive, and equivalent to bounded width.

The rest of this section is devoted to the proof of Theorem~\ref{thm:bw2sw2}:
assuming a coset template $T$ is not 2-Helly, we construct a family of instances that
are hard for every local equivariant algorithm.
Interestingly, the hard instances are a generalization of CFI graphs~\cite{CFI92}.
The idea of the proof generalizes the construction of~\cite{BKLT13}.

%Theorem~\ref{thm:bw2sw2} is not a new result: it follows from
%Lemma~9 in~\cite{DL08} and from the results of~\cite{LZ07}.
%On the other hand, the proof of Theorem~\ref{thm:bw2sw2} is not at all specific for the consistency algorithm, 
%and can be straightforwardly adapted to any other algorithm satisfying the following two properties.

\begin{proof}
Fix a coset template $T$ being not 2-Helly, and a local equivariant algorithm. We aim at showing that the algorithm does not correctly solve CSP($T$). 
Let $\data_i(I)$ denote the object computed in the $i$th stage of the algorithm on input $I$. 
% suppose it is not 2-Helly. We aim at showing that $T$ has no bounded width.
%For the ease of reading we state some claims without proofs. The proofs are provided afterwords.

We start with the following claim, whose proof is postponed to Section~\ref{sec:almost direct product}:

\begin{proposition} \label{prop:almost direct product}
There are some subgroups $S_1, S_2, S_3$ of carrier groups, and an almost-direct product
\[
R \leq S_1 \times S_2 \times S_3
\]
such that $R$ and all its cosets in $S_1 \times S_2 \times S_3$ are pp-definable (as ternary relations).
\end{proposition}

%
%such that $S_1, S_2, S_3$ and $R$ are all pp-definable (as unary or, respectively, ternary relations in $T$) 
%and witness the following properties:
%\begin{align} 
%\label{eq:stronganomaly1}
%%\hspace{-0.5cm}
%%\exists \pi_1 \in S_1, \ \pi_2 \in S_2, \ \pi_3 \in S_3 . \ \neg R(\pi_1, \pi_2, \pi_3) \\
%R \neq S_1 \times S_2 \times S_3  \\
%\label{eq:stronganomaly2}
%\forall \pi_2 \in S_2, \ \pi_3 \in S_3, \ \exists \pi_1 \in S_1 . \ R(\pi_1, \pi_2, \pi_3) \\
%\label{eq:stronganomaly3}
%\forall \pi_1 \in S_1, \ \pi_3 \in S_3, \ \exists \pi_2 \in S_2 . \ R(\pi_1, \pi_2, \pi_3) \\
%\label{eq:stronganomaly4}
%\forall \pi_1 \in S_1, \ \pi_2 \in S_2, \ \exists \pi_3 \in S_3 . \ R(\pi_1, \pi_2, \pi_3) 
%\end{align}

Following the idea of~\cite{BKLT13},
we will now define  a class of  instances, called $n$-torus instances, and then show that the consistency algorithm yields incorrect results for these instances.
An $n$-torus instance is an instance of particular shape. It contains exactly $3 n^2$ elements
\[
a_{i j}, \ b_{i j}, \ c_{i j},
\quad \text{ for } i, j \in \{0 \ldots n-1\},
\]
and exactly $2 n^2$ constraints:
\begin{align} \label{eq:constraints}
R_{i j}(a_{i j}, b_{i j}, c_{i j}) \ \text{ and } \  R'_{i j}(a_{i (j{+}1)}, b_{(i{+}1) j}, c_{i j}),  
\end{align}
for $i, j \in \{0 \ldots n-1\}$.
We adopt the convention that all indices are counted modulo $n$, e.g.~$a_{i n} = a_{i 0}$ and $a_{n i} = a_{0 i}$.
Relations $R_{i j}$ and $R'_{i j}$ are arbitrary cosets of $R$ in $S_1 \times S_2 \times S_3$;
by Proposition~\ref{prop:almost direct product} they are all pp-definable.
% Fact~\ref{fact:pp-definable equivalent}.
% \slcomm{ze wolno uzywac pp-definiowalnych}
Formally, the constraints in $n$-torus instance are not just relations from $T$, but rather pp-definable relations in $T$. 
We rely here on a folklore common knowledge: a pp-definable constraint can be simulated by adjoining to an instance a gadget, whose size is the number of (existentially) quantified variables in the defining pp-formula.

The $2 n^2$ tuples 
$
(a_{i j}, b_{i j}, c_{i j}) \ \text{ and } \ (a_{i (j{+}1)}, b_{(i{+}1) j}, c_{i j})
$
appearing in constraints~\eqref{eq:constraints} we call \emph{positions} of an $n$-torus.
The shape of
a $3$-torus instance is depicted below, with triangles representing positions and
sides of triangles representing elements.
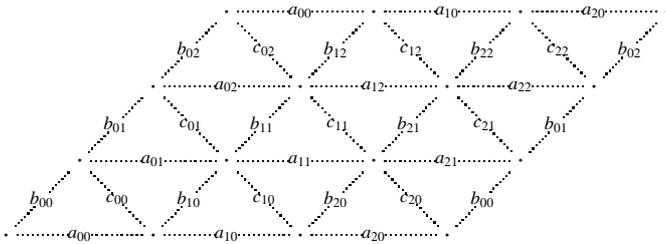
\[
\xymatrix@-5pt{
&&& \cdot \ar@{.}[rr]|-{a_{0 0}} && \cdot \ar@{.}[rr]|-{a_{1 0}} && \cdot \ar@{.}[rr]|-{a_{2 0}} && \cdot \\
&& \cdot \ar@{.}[rr]|-{a_{0 2}}  \ar@{.}[ru]|-{b_{0 2}} && 
\cdot \ar@{.}[rr]|-{a_{1 2}} \ar@{.}[ru]|-{b_{1 2}}  \ar@{.}[lu]|-{c_{0 2}}  && 
\cdot \ar@{.}[rr]|-{a_{2 2}} \ar@{.}[ru]|-{b_{2 2}}  \ar@{.}[lu]|-{c_{1 2}} && 
\cdot \ar@{.}[ru]|-{b_{0 2}}  \ar@{.}[lu]|-{c_{2 2}}  &  \\
& \cdot \ar@{.}[rr]|-{a_{0 1}}  \ar@{.}[ru]|-{b_{0 1}} && 
\cdot \ar@{.}[rr]|-{a_{1 1}}   \ar@{.}[ru]|-{b_{1 1}}  \ar@{.}[lu]|-{c_{0 1}} && 
\cdot \ar@{.}[rr]|-{a_{2 1}}   \ar@{.}[ru]|-{b_{2 1}}  \ar@{.}[lu]|-{c_{1 1}} && 
\cdot  \ar@{.}[ru]|-{b_{0 1}}  \ar@{.}[lu]|-{c_{2 1}}  &&    \\
\cdot \ar@{.}[rr]|-{a_{0 0}}   \ar@{.}[ru]|-{b_{0 0}} && 
\cdot  \ar@{.}[rr]|-{a_{1 0}}  \ar@{.}[ru]|-{b_{1 0}}  \ar@{.}[lu]|-{c_{0 0}}  && 
\cdot  \ar@{.}[rr]|-{a_{2 0}}   \ar@{.}[ru]|-{b_{2 0}}  \ar@{.}[lu]|-{c_{1 0}}  && 
\cdot   \ar@{.}[ru]|-{b_{0 0}}  \ar@{.}[lu]|-{c_{2 0}}   && 
}
\]

%\[
%\xymatrix@-2pc{
%&&& a_{01} \ar@/^/@{-}[rdd] &&&& a_{11} &&&& a_{21} \\
%\\
% b_{00}  \ar@/^1pc/@{.}[rr] & & c_{00}  \ar@/^/@{.}[ldd] \ar@/^/@{-}[ruu] && b_{10} \ar@/^1pc/@{-}[ll] && c_{10} && b_{20} && c_{20} \\
% \\
%& a_{00} \ar@/^/@{.}[luu] &&&& a_{10} &&&& a_{20} &&&&
%}
%\]

%\begin{figure}[htbp]
%	\centering
%		\includegraphics[scale=0.30]{torus}
%	\caption{An $n$-torus.  }
%	\label{fig:torus}
%\end{figure}

Every element $a$ of an $n$-torus appears in exactly two constraints.
Thus every position $(a, b, c)$ has exactly three \emph{neighbors}, 
namely those other positions that contain any of $a, b, c$. 
For instance, neighbors of the position $(a_{12}, b_{21}, c_{11})$ are
\[
(a_{12}, b_{12}, c_{12}), \quad
(a_{21}, b_{21}, c_{21}), \quad \text{and }\ 
(a_{11}, b_{11}, c_{11}).
\]
This defines the 3-regular neighborhood graph, with vertices being the positions of an $n$-torus.

The $n$-torus instances are built from triangulations of a torus surface.
It is however not particularly important to use a torus; equally well a sphere could be used instead, 
or any other connected closed surface, as long as, intuitively speaking,
% the following easy observation is satisfied.Intuitively, the observation says that an $n$-torus
the surface is hard to decompose into small pieces.
The non-decomposability can be formally stated as follows:
\begin{fact}\label{lem:small-triangulation}
After removing  $j < n$ positions, the neighborhood graph of an $n$-torus still contains a connected component of size 
at least $2n^2-j^2$.
\end{fact}
Indeed, locally, the neighborhood graph of an $n$-torus can be seen as a 3-regular graph on the plane.
Thus, in order to isolate $j^2$ positions one needs to cut more than $j$ edges.

%Consider the $(k, l)$-consistency algorithm, for fixed $k \leq l$.

\begin{definition}
Let $\torus$ be an $n$-torus, and $i \geq 0$.
We say that the % $(k, l)$-consistency 
algorithm \emph{ignores}  a position $(a, b, c)$ of $\torus$ % at $X$ 
after the $i$th stage, if
\[
\data_i(\torus) \ = \ \data_i(\torus')
%{\cal H}_X({\torus}, i) \ = \ {\cal H}_X({\torus'}, i),
\]
for every $n$-torus $\torus'$ that differs from $\torus$ only by one constraint at position $(a, b, c)$.
\end{definition}
Using Fact~\ref{lem:small-triangulation} we prove:
\begin{proposition} \label{prop:does not observe}
There is some $m \in \Nat$
% depending only on $l$ 
such that for sufficiently large $n$ and every $n$-torus $\torus$, 
the 
% $(k,l)$-consistency 
algorithm ignores, 
% at every $X$ and 
after every stage, all but at most $m$ positions of $\torus$.
\end{proposition}
Therefore for every sufficiently large instance, 
the 
% $(k, l)$-consistency 
algorithm necessarily ignores some position after the last stage, which 
easily entails incorrectness of the algorithm. Indeed, consider an $n$-torus $\torus_R$
that uses the relation $R$ of Proposition~\ref{prop:almost direct product}
in all its constraints. Being a subgroup instance, $\torus_R$ is solvable 
%(for instance, use the same arbitrarily chosen triple $(a, b, c)\in R$ to satisfy all constraints) 
and hence the % $(k,l)$-consistency 
algorithm answers positively. 
On one hand, by Proposition~\ref{prop:does not observe}
there is some position $(a_0, b_0, c_0)$ such that the 
% $(k,l)$-consistency 
algorithm answers positively
for the instance obtained by replacing the relation $R$ in the constraint $R(a_0, b_0, c_0)$ in ${\torus}_R$
with any other coset of $R$ in $S_1 \times S_2 \times S_3$.
On the other hand, we prove:
\begin{proposition} \label{prop:unsolvable}
Replacing the relation $R$ with any other coset of $R$
in one constraint % $R(a_0, b_0, c_0)$ 
in  ${\torus}_R$
% in $S_1 \times S_2 \times S_3$ 
yields an unsolvable instance.
\end{proposition}
In consequence, the 
% $(k, l)$-consistency 
algorithm is incorrect.
This completes the proof of Theorem~\ref{thm:bw2sw2}, once
%
%\begin{lemma} \label{lem:bw}
%If $T$ has a template anomaly then $T$ has no bounded width.
%\end{lemma}
%
%The lemmas imply Lemma~\ref{lem:bw2helly}.
%Indeed, due to Lemma~\ref{lem:anomaly} some instance of $T$ contains necessarily an anomaly of size $3$, hence by Lemma~\ref{lem:stronganomaly}
%$T$ has a template anomaly. Then by Lemma~\ref{lem:bw} we deduce that $T$ has no bounded width, as required.
%
%\end{proof}
we prove the three yet unproved claims, namely Propositions~\ref{prop:almost direct product}, \ref{prop:does not observe}
and~\ref{prop:unsolvable}.
\end{proof}

\subsection{Proof of Proposition~\ref{prop:almost direct product}}
\label{sec:almost direct product}

By Lemma~\ref{lem:iffanomaly} some instance $I$ contains a $(k,k+1)$-anomaly, for $k \geq 2$. 
Note that this implies that this instance has at least one solution.

Wlog.~we can assume that $I$ is a subgroup instance.
Indeed, for an arbitrary solution $h$ of $I$, define a new instance by the action of $h^{-1}$:
\[
I' \ := \ \ I \cdot h^{-1}.
\]
As $h$ is a solution of $I$, for every constraint $R(a_1, \ldots, a_n)$ in $I$, the tuple $(h(a_1), \ldots, h(a_n))$ 
is in $R$.
Hence every relation appearing in a constraint of $I'$ is a subgroup in the product of constraining groups, as required.
Due to Fact~\ref{fact:action solution} an anomaly admitted by $I$ translates to an anomaly admitted by $I'$.

By Lemma~\ref{lem:anomaly} we deduce that some (possibly different) instance $I$ admits a $(2, 3)$-anomaly  $h = (\pi_1, \pi_2, \pi_3)$.
Consider the set $H$ of all those partial solutions, with the same domain as $h$, that extend to a solution of $I$.
$H$ is a pp-definable ternary relation according to Fact~\ref{fact:solutions pp-definable}.
By Fact~\ref{fact:solutions coset}(2) we know that $H$
is a subgroup in the product $G_1 \times G_2 \times G_3$ of some three carrier groups.
As $h$ is a $(2,3)$-anomaly, we know
(we prefer below to write $H(\pi_1, \pi_2, \pi_3)$ instead of $(\pi_1, \pi_2, \pi_3) \in H$):
\begin{align} \label{eq:weakanomaly1}
\neg H(\pi_1, \pi_2, \pi_3) \\
\label{eq:weakanomaly2} 
\exists \tau \in G_1 . \ H(\tau, \pi_2, \pi_3) \\
\label{eq:weakanomaly3} 
\exists \tau \in G_2 . \ H(\pi_1, \tau, \pi_3) \\
\label{eq:weakanomaly4}
\exists \tau \in G_3 . \ H(\pi_1, \pi_2, \tau) 
\end{align}

Now we are ready to define an almost-direct product $R \leq S_1 \times S_2 \times S_3$.
The subgroups $S_1 \leq G_1$, $S_2 \leq G_2$ and $S_3 \leq G_3$ we define as follows:
\begin{align*}
\tau_1 \in S_1 \quad \iff \quad \exists \tau. \ H(\tau_1, \tau, 1) \ \land \ \exists \tau. \ H(\tau_1, 1, \tau) \\
\tau_2 \in S_2 \quad \iff \quad \exists \tau. \ H(\tau, \tau_2, 1) \ \land \ \exists \tau. \ H(1, \tau_2, \tau) \\
\tau_3 \in S_3 \quad \iff \quad \exists \tau. \ H(\tau, 1, \tau_3) \ \land \ \exists \tau. \ H(1, \tau, \tau_3) 
\end{align*}
and the subgroup $R$ as the restriction of $H$ to $S_1 \times S_2 \times S_3$:
\begin{align*}
& R \ := \ \ H \ \cap \ S_1 \times S_2 \times S_3.
\end{align*}
By the very definition, $S_1, S_2, S_3$ and $R$ are pp-definable.

We need to show the conditions \eqref{eq:stronganomaly1}--\eqref{eq:stronganomaly4} in Definition~\ref{def:almost direct product}.
For~\eqref{eq:stronganomaly1} (i.e.~$R \neq S_1 \times S_2 \times S_3$) we use~\eqref{eq:weakanomaly1} and~\eqref{eq:weakanomaly2} to conclude that
for $\tau_1 = \tau^{-1} \pi_1 \in G_1$ it holds
\[
\neg H(\tau_1, 1, 1).
\] 
Moreover, using~\eqref{eq:weakanomaly2} and~\eqref{eq:weakanomaly3} we deduce that for some $\bar\tau \in G_2$
\begin{align}  \label{e-rev1}
H(\tau_1, \bar\tau, 1);
\end{align}
and similarly, using~\eqref{eq:weakanomaly2} and~\eqref{eq:weakanomaly4}, we deduce that for some $\bar\tau \in G_3$,
\[
H(\tau_1, 1, \bar\tau).
\]
Thus $\tau_1 \in S_1$ and therefore $(\tau_1, 1, 1) \in S_1 \times S_2 \times S_3 \setminus R$.

Now we concentrate on condition~\eqref{eq:stronganomaly4}  in Definition~\ref{def:almost direct product}
(the remaining two conditions~\eqref{eq:stronganomaly2} and~\eqref{eq:stronganomaly3} 
are shown in the same way).
Let $\tau_1 \in S_1$ and $\tau_2 \in S_2$. By the very definition of $S_1$ and $S_2$ we learn
\begin{align}
\begin{aligned} \label{e-rev2} H(\tau_1, 1, \tau) \end{aligned} \\ %\qquad \text{and} \qquad \\
\begin{aligned} \label{e-rev3} H(1, \tau_2, \tau') \end{aligned}
\end{align}
for some $\tau, \tau'  \in G_3$. % and thus $\tau \tau' \in G_3$. 
Therefore $H(\tau_1, \tau_2, \tau \tau')$ and it only remains to show that
$\tau \tau' \in S_3$. Consider $\tau$ ($\tau'$ is treated analogously) in order to show $\tau \in S_3$. 
The fact~\eqref{e-rev2} proves a half of the defining condition for $\tau \in S_3$, while the other half
is proved by combining~\eqref{e-rev2} with~\eqref{e-rev1} to deduce
$
H(1, \bar\tau^{-1}, \tau).
$
We have thus shown that $R$ is an almost-direct product of $S_1$, $S_2$, $S_3$.
% This completes the proof of $\tau \in S_3$.
%
%
%Finally we concern pp-definability. $H$, being a projection of the set of all solutions of an instance $I$, is a pp-definable ternary relation.
%Then, the very definition of $S_1, S_2, S_3$ and $R$ proves them to be pp-definable as well, as unary and ternary relations respectively, 
%because pp-definable relations are closed
%under existential quantification and conjunction. % (projection and intersection.

Fact~\ref{fact:action ppdef}(2) guarantees that all cosets of $R$ are pp-definable, as required.

\subsection{Proof of Proposition~\ref{prop:does not observe}}

We will need the following property of almost-direct products:
\begin{lemma} \label{lem:1coord}
Every coset $R'$ of $R$ in $S_1 \times S_2 \times S_3$ contains elements of the form
\[
(\tau_1, 1, 1), \quad (1, \tau_2, 1), \quad (1, 1, \tau_3),
\]
for some $\tau_1 \in S_1, \tau_2 \in S_2$ and $\tau_3 \in S_3$.
\end{lemma}
\begin{proof}
Indeed, let $ \pi = (\pi_1, \pi_2, \pi_3) \in R'$. Knowing that $\rho = (\rho_1, \pi_2, \pi_3) \in R$ for some $\rho_1 \in S_1$, we get
\[
\rho^{-1} \pi \ = \ (\rho_1^{-1} \pi_1, 1, 1) \in R'
\]
as required. Likewise one proves the remaining two claims.
\end{proof}

Let $m = (2d)^2$, where $d$ is the locality bound of the algorithm, and 
let $\torus$ be an $n$-torus, for $n$ sufficiently large to satisfy $2n^2 > (d+1) \cdot m$.

We proceed by induction on the number $i$ of a stage.
Observe that elements of a set $\supp \subseteq I$ of size at most $d$ appear in at most $2d$ positions of $I$.
Therefore, whenever (L0) applies (this is necessarily the case when $i = 1$),
all but at most $2d < m$ positions are ignored after the $i$th stage.

Otherwise, suppose (L1) and (L2) hold. Let $\supp \subseteq I$ be such that
\begin{align}  \label{eq:zero}
\data_i(I) \cdot s \ = \ \data_i(I) \quad \text{ for all } s \text{ fixing } \supp.
\end{align}
%%%%Now we prove the induction step, $i > 1$. 
%%%recall that $\data_i(I) \cdot s = \data_i(I)$  for all pre-solutions that satisfy $s(x) = 1$ for all $x\in \supp_i$.
%%By the locality of the algorithm there is a number $d$, not depending on $i$, such that the value
%%$\data_i$ depends only on some $d$ among values $\data_1, \ldots, \data_{i-1}$. 
% We only need to consider the unique set ${\cal H}_X({\torus}, i)$ that changes in this stage.
%According to the consistency condition tested by the algorithm, 
%there is a set $Y$ of $l$ elements of ${\torus}$ such that the new value of ${\cal H}_X$ depends only on
%$o(k, l)$ values of ${\cal H}_{X'}$ for $X' \subseteq Y$ (cf.~\eqref{eq:depends on}). 
By (L2) we can assume that the size of $\supp$ is at most $d$.
By the induction assumption for the previous stages and by (L1),
%and all sets $X'$, 
there are at most $d \cdot m$
%$o(k, l) \cdot m$ constraints 
positions not ignored after stage $i$. % at $X$.
We need to show, however, that there are at most $m$ such positions.

The argument has a geometric flavor, and builds on Fact~\ref{lem:small-triangulation}:
%Define $j = 2 k$, the maximal number of constraints referring to elements of $X$. We now reveal the value of $m$; we put $m=j^2$.
after removing from the neighborhood graph all positions in which elements of $\supp$ appear (there is at most $2d$ of them),
there is a connected subset $C$ of positions of size at least $2n^2-m$, so it is larger than $d \cdot m$.
By the induction assumption, some position in $C$ is ignored after the $i$th stage.
For the proof of Proposition~\ref{prop:does not observe}
it is enough to prove that \emph{every} position in $C$ is ignored after $i$th stage.
To this end, since $C$ is connected, it is now enough to show:
\begin{claim}
If some position in $C$ is ignored after the $i$th stage, then every neighbor of that position in $C$ also is.
\end{claim}
To show the last claim, 
consider two neighboring constraints in $C$, say $U(a, b, c)$ and $U'(a, b', c')$, both referring to an element $a$.
Supposing that $(a, b, c)$ is ignored, we need to demonstrate that $(a, b', c')$ is ignored too.
Let $\overline \torus'$ be an $n$-torus obtained 
from $\torus$ by replacing the constraint $U'(a, b', c')$ with $\overline U'(a, b', c')$, for some coset
$\overline U' \ = \ U' \pi.$ We need to show 
\begin{align} \label{eq:0}
%{\cal H}_X({\torus}, i) \ = \ {\cal H}_X({\overline \torus'}, i).
\data_i({\torus}) \ = \ \data_i({\overline \torus'}).
\end{align}
Using Lemma~\ref{lem:1coord} we may assume wlog.~that $\pi = (\pi_1, 1, 1)$ for some $\pi_1 \in S_1$.
Let $s$ be a pre-solution defined by
\[
s(x) = \begin{cases} \pi_1 & \text{if } x = a\\
1 & \text{otherwise.}
\end{cases}
\] 
Knowing that $(a, b, c)$ is ignored, we may write
\begin{align} \label{eq:1}
%{\cal H}_X({\torus}, i) \ & = \ {\cal H}_X({\overline \torus}, i),
\data_i({\torus}) \ & = \ \data_i({\overline \torus}),
\end{align}
where the $n$-torus $\overline \torus$ is obtained from $\torus$ by replacing the constraint $U(a, b, c)$ with 
$\overline U(a, b, c)$, for $\overline U = U \pi^{-1}.$
Observe the equality
\begin{align}  \label{eq:s}
{\overline \torus'} \ = \ {\overline \torus} \cdot s.
\end{align}
%As $s(x) = 1$ for all $x \in \supp_i$, by (L2) we deduce
%\begin{align} \label{eq:2}
%\data_i({\overline \torus}) \ = \ \data_i({\overline \torus'}).
%\end{align}
%
Now we are ready to prove~\eqref{eq:0} by composing the following equalities:
\begin{align*} 
%{\cal H}_X({\torus}, i) \ & = \ {\cal H}_X({\torus} \cdot s, i) \\
%{\cal H}_X({\torus}\cdot s, i) \ & = \ {\cal H}_X({\overline \torus} \cdot s, i).
\data_i({\torus}) \ & = \ \data_i({\torus}) \cdot s \  = \ \data_i({\torus} \cdot s) 
 \ = \ \data_i({\overline \torus} \cdot s)  \ = \ \data_i({\overline \torus'}).
\end{align*}
The first equality follows by \eqref{eq:zero}, as $s$ fixes $\supp$;
the second one is the equivariance condition; the third one follows by~\eqref{eq:1} combined with equivariance;
and the last one is a consequence of~\eqref{eq:s}.

\subsection{Proof of Proposition~\ref{prop:unsolvable}}

Fix a position $(a_0, b_0, c_0)$.
Let ${\torus}^{-}_R$ be the instance obtained from ${\torus}_R$ by removing the constraint 
$R(a_0, b_0, c_0)$.
We will show that every solution $h$ of ${\torus}^-_R$ satisfies the constraint $R(a_0, b_0, c_0)$:
%i.e.
\begin{align} \label{eq:want}
% h(a_0, b_0, c_0) % =  
(h(a_0), h(b_0), h(c_0)) 
\in R.
\end{align}

%For conciseness,
%% For a position $(a, b, c)$ and a pre-solution $h$ of ${\torus}$ we write 
%instead of $(h(a), h(b), h(c))$  we write $h(a, b, c)$.  
% \in S_1 \times S_2 \times S_3$.
According to the definition of $n$-torus,   %~\eqref{eq:constraints} 
the positions of $\torus_R$ split into two disjoint subsets, call them 
\emph{negative} and \emph{positive},  so that neighbors of a negative position
are positive, and vice versa.
Wlog.~assume that $(a_0, b_0, c_0)$ is negative.
%:
Consider the following expression (symbol $\prod$ stands for the group operation in $R$,
applied in an unspecified but irrelevant order):
\begin{align} \label{eq:expre}
\prod_{(a, b,c) \text{ negative}} \hspace{-2mm} (h(a), h(b), h(c))^{-1} \  \prod_{(a, b, c) \text{ positive}} \hspace{-2mm} (h(a), h(b), h(c)),
\end{align}
where $(a, b, c)$ in the first subexpression ranges over all negative positions of ${\torus}^-_R$
(hence $(a_0, b_0, c_0)$ is omitted), and in the second subexpression over all positive ones.
The expression~\eqref{eq:expre} evaluates to some value $(\pi_1, \pi_2, \pi_3)$ in $R$.
% ; the value of the expression may however,
% in principle, depend on the order of application of the group operation.
% We aim at showing that irrespectively of the order of application of the group operation,
% the expression~\eqref{eq:expre} evaluates to $h(a_0, b_0, c_0)$, which implies~\eqref{eq:want}.

Let $\coset{\_} : R \to \coset{R}$ be a surjective group homomorphism from $R$ to a 
commutative group $\coset{R}$, 
%strict almost-direct product 
%\[
%\coset{R} \leq \coset{S_1} \times \coset{S_2} \times \coset{S_3},
%\]
guaranteed jointly by Lemmas~\ref{lem:strictquotient} and~\ref{lem:abelian}.
Recall from Section~\ref{sec:almost} that the homomorphism $\coset{\_}$ is defined point-wise, namely 
$\coset{(\tau_1, \tau_2, \tau_3)} = (\coset{\tau_1}, \coset{\tau_2}, \coset{\tau_3})$.
Apply $\coset{\_}$ to~\eqref{eq:expre} to get an expression:
\begin{align} \label{eq:expre f}
%f(\pi) \ = \ 
\prod_{(a, b,c) \text{ negative}} \hspace{-3mm} (\coset{h(a)}, \coset{h(b)}, \coset{h(c)})^{-1}  \ 
\prod_{(a, b, c) \text{ positive}} \hspace{-3mm} (\coset{h(a)}, \coset{h(b)}, \coset{h(c)}).
\end{align}
% which, irrespectively of the order of application of the group operation, evaluates to the same value in $\coset{R}$.
%
Observe that $\coset{h(a)}$ appears in~\eqref{eq:expre f}
exactly once, for every $a \in I$ different from $a_0, b_0, c_0$; the same applies to the inverse $\coset{h(a)}^{-1}$. 
Thus, as $\coset{R}$ is commutative, every $\coset{h(a)}$ together with its inverse cancels out.
Moreover, $\coset{h(a_0)}$, $\coset{h(b_0)}$ and $\coset{h(c_0)}$ also appear in~\eqref{eq:expre f} exactly once,
while their inverses do not appear as the negative position $(a_0, b_0, c_0)$ has been omitted.
Therefore the expression~\eqref{eq:expre f} evaluates to $(\coset{h(a_0)}, \coset{h(b_0)}, \coset{h(c_0)})$.
On the other hand, \eqref{eq:expre f} necessarily evaluates to $(\coset{\pi_1}, \coset{\pi_2}, \coset{\pi_3})$.
% , thus $(\coset{h(a_0)}, \coset{h(b_0)}, \coset{h(c_0)})  = (\coset{\pi_1}, \coset{\pi_2}, \coset{\pi_3})$.
Using the notation of Section~\ref{sec:almost} we can write:
\[
h(a_0) \equiv_1 \pi_1 \quad h(b_0) \equiv_2 \pi_2 \quad h(c_0) \equiv_3 \pi_3. 
\]
Now we use the closure of $R$ on the equivalences, cf.~\eqref{eq:equiv} in Section~\ref{sec:almost}:
as $(\pi_1, \pi_2, \pi_3) \in R$, we deduce $(h(a_0), h(b_0), h(c_0)) \in R$ as required.
%; as $f$ is a group homomorphism  we deduce that the expression~\eqref{eq:expre} evaluates to $h(a_0, b_0, c_0)$, as required.
Proposition~\ref{prop:unsolvable} is thus proved.

\ \\
\begin{remark}
Splitting the positions into positive and negative ones, with one more positive than negative ones, resembles the property of ability to count of~\cite{FV98}.
We believe that the proof can be modified to prove the equivalence: a coset template is not 2-Helly if and only if some its pp-definable extension has the ability to count. The latter property needs however to be slightly generalized to work in our setting, as the setting allows many different carrier groups. The equivalence is not a new result: it been shown recently for all templates in~\cite{BK14} .
\end{remark}

\paragraph{Acknowledgements}
The author is grateful to Szymek Toru{\'n}czyk for proposing a simplified proof of Lemma~\ref{lem:abelian}.
Moreover, thanks go to Bartek Klin, Ania Ochremiak, and Szymek Toru{\'n}czyk for long and
fascinating discussions on CSP and its relationship to computation in sets with atoms.
The author thanks also Luc Segoufin for encouraging me to write down this note.
Finally, thanks go to the anonymous reviewers for their insightful and helpful remarks.
\ \\



%\bibliographystyle{elsarticle-num} 
\bibliographystyle{plain} 
\bibliography{bib}

\end{document}